\documentclass[journal]{IEEEtran}
\usepackage{graphicx}
\usepackage{amsmath,amsfonts,amssymb}
\usepackage{graphics}
\usepackage{epsfig}
\usepackage{footmisc}
\usepackage{fixfoot}
\usepackage{balance}

\newtheorem{theorem}{Theorem}
\newtheorem{lemma}{Lemma}

\newtheorem{proof}{Proof}
\IEEEoverridecommandlockouts

\begin{document}

\title{When mmWave Communications Meet Network Densification: A Scalable Interference Coordination Perspective}

\author{Wei~Feng,~\IEEEmembership{Member,~IEEE,}
        Yanmin~Wang,
        Dengsheng~Lin,~\IEEEmembership{Member,~IEEE,}
        Ning~Ge,~\IEEEmembership{Member,~IEEE,}
        Jianhua~Lu,~\IEEEmembership{Fellow,~IEEE,}
        and~Shaoqian~Li,~\IEEEmembership{Fellow,~IEEE}
\thanks{This work was supported in part by the National Basic Research Program of China under grant No. 2013CB329001, and
the National Science Foundation of China under grant No.~91638205 and grant No.~61621091.}
\thanks{W. Feng, N. Ge, and J. Lu are with the Tsinghua National Laboratory
for Information Science and Technology, Tsinghua University, Beijing, China~(email:
fengw@mails.tsinghua.edu.cn, gening@tsinghua.edu.cn, lhh-dee@tsinghua.edu.cn).}
\thanks{Y. Wang is with the China Academy of Electronics and Information Technology, Beijing, China~(email:
yanmin-226@163.com).}
\thanks{D. Lin and S. Li are with the National Key Laboratory of Science
and Technology on Communications, University of Electronic Science and
Technology of China, Sichuan, China~(email: linds@uestc.edu.cn,
lsq@uestc.edu.cn).}
\thanks{Digital Object Identifier XXX.}
}

\markboth{IEEE JOURNAL ON SELECTED AREAS IN COMMUNICATIONS,~Vol.~xx, No.~x, 2017}
{Feng \MakeLowercase{\textit{et al.}}: When mmWave Communications Meet Network Densification: A Scalable Interference Coordination Perspective}

\maketitle

\begin{abstract}
The millimeter-wave (mmWave) communication is envisioned to provide orders of magnitude capacity improvement.
However, it is challenging to realize a sufficient link margin due to high path loss and blockages.
To address this difficulty, in this paper, we explore the potential gain of ultra-densification for enhancing mmWave communications from a network-level perspective.
By deploying the mmWave base stations (BSs) in an extremely dense and amorphous fashion, the access distance is reduced and
the choice of serving BSs is enriched for each user, which are intuitively effective
for mitigating the propagation loss and blockages. Nevertheless, co-channel interference under this model will
become a performance-limiting factor. To solve this problem, we propose a
large-scale channel state information (CSI) based interference coordination approach. Note that
the large-scale CSI is highly location-dependent, and can be obtained with a quite low cost. Thus, the
scalability of the proposed coordination framework can be guaranteed.
Particularly, using only the large-scale CSI of interference links, a coordinated frequency resource block allocation problem is formulated for maximizing the minimum achievable rate of the users, which is uncovered to be a NP-hard integer programming problem.
To circumvent this difficulty, a greedy
scheme with polynomial-time complexity is proposed by adopting the bisection method and linear integer programming tools.
Simulation results demonstrate that
the proposed coordination scheme based on large-scale CSI only can still offer
substantial gains over the existing methods.
Moreover, although the proposed scheme is only guaranteed to converge to a local optimum,
it performs well in terms of both user fairness and system efficiency.
\end{abstract}


\section{Introduction}
\IEEEPARstart{T}{he} future generation (5G and beyond) wireless communication system is envisioned to serve a massive amount of users, while satisfying
various demands of multimedia services. Towards this end, a lot of emerging or evolutionary technologies have been
extensively studied~\cite{r01}. In particular, the millimeter-wave (mmWave) communication
has been recognized as a promising technology for providing orders of magnitude capacity improvement~\cite{r02}-\cite{rxiao2}, due to the large bandwidth
available at mmWave bands.

The key difficulty of mmWave communications is the severe propagation attenuation caused by
high path loss, shadowing and blockages~\cite{r03}-\cite{rxiao3}. In practical applications, the transmit power is usually quite limited, due to the hardware constraint of mmWave transceivers~\cite{r04}, which further renders it challenging to realize a sufficient link margin~\cite{r05}.

To deal with this problem, most of the existing research efforts have been devoted to link-level
directional beamforming design~\cite{r06}-\cite{r010}, aiming at energy-focusing wireless transmission.
This approach is valid, as it becomes feasible to adopt large antenna arrays at the mmWave transceivers,
thanks to small wavelengths of mmWave signals. Further towards practice,
by considering the combination of analog and digital processing with a limited number of
radio frequency (RF) chains, the total power consumption as well as the hardware cost can be significantly reduced.
Specifically, in~\cite{r06}, an iterative hybrid analog/digital beamforming algorithm was proposed for the single-user mmWave channel.
The same authors further developed an adaptive algorithm
to estimate the mmWave channel via designing a hierarchical multi-resolution
codebook. Using the estimated channel, a hybrid precoding algorithm
was proposed in~\cite{r07}. In~\cite{r08}, by adopting the theory of matrices, a joint RF and baseband precoder was contrived, which has low
computational complexity, and enables the parallel hardware implementation.
To uncover the optimal design of hybrid precoders, the authors of~\cite{r09} formulated
the precoder optimization problem as a matrix factorization problem, and an
innovative design methodology was proposed based on the idea of alternating minimization. Further considering the generalized frequency-selective channel,
the hybrid analog/digital codebooks were developed for wideband mmWave systems in~\cite{r010}.

The aforementioned link-level results are quite insightful, however, we still face challenges.
First, the mmWave channel is measured to be sparse in terms of multi-path components~\cite{r011}, which would inevitably
lead to small degrees of freedom of the link, thus limiting the performance of multi-antenna precoding.
Second, although the antenna array can be tightly packed at the transceiver,
real-time beam training is crucially required to gain the benefit of antenna arrays, which may
cause an unaffordable cost at user equipments and greatly limit user mobility~\cite{r04}. Third,
mmWave communications are quite sensitive to blockages~\cite{r04}, however, link-level precoding
makes little sense to the severe outage resulted by a stubborn blockage~\cite{r016}.

To address these challenges, in this paper, we promote mmWave communications from a network-level perspective.
Particularly, we consider an extremely dense and amorphous mmWave network~\cite{r014}\cite{r015}, where a large amount of mmWave base stations (BSs)
are randomly deployed in an amorphous way in the coverage area. This assumption follows the fact that
the deployment of multiple mmWave BSs will not be cellular-like regular due to limited sites\footnote{Taking the
maritime mmWave communication as an example, for which the BSs can only be deployed along the coast or on an available island.} and the target towards mitigation of blockages.
Furthermore, some mmWave BSs may be randomly deployed by subscribers~\cite{r017}, which also leads to an amorphous system topology.
The key advantage of this model lies in the greatly reduced access distance and enriched choice of serving BSs for each user~\cite{r3}-\cite{rfeng1},
which intuitively are more straightforward pathways to mitigate the propagation loss and blockages.
For instance, the authors of~\cite{r012} have shown that the cooperation among multiple mmWave transmitters is effective for
alleviating the link fluctuation problem. Coming to cellular mmWave communications, an enhanced coverage performance was
identified in recent studies~\cite{r016}\cite{r013}.
Thanks to this superiority, we further liberate the BSs and user equipments from the heavy cost of beam training. To be specific, only a single RF chain and only the analog beamforming, i.e., phase shifting, are assumed at the BS.
Meanwhile, only a single omnidirectional antenna is equipped at the mobile terminals. These assumptions are more practical
than the existing studies~\cite{r06}-\cite{r010} for the future mobile mmWave communication network.

Under this model, the co-channel interference is a critical performance-limiting factor~\cite{r04}.
Different from traditional studies that
use full channel state information (CSI) for interference management~\cite{r012},~\cite{r8}-\cite{r16}, we assume
only the slowly varying large-scale CSI is available for each interference link~\cite{rfeng2}\cite{rfeng3}.
Note that the acquisition of full CSI will lead to a huge cost, including the pilot
overhead, the extra processor for interference-channel estimation, as well as the feedback
overhead for frequency division duplex regimes, which makes it not viable for practical applications~\cite{r18}-\cite{r17}.
In a sharp contrast, the large-scale CSI is highly location-dependent, and can be obtained
with a quite low cost~\cite{r016}. Therefore, the scalability of interference coordination
can be guaranteed, which is crucial for the considered mmWave network. Intuitively, due to the amorphous
system topology, it is difficult to decouple the considered system with an acceptable performance loss\footnote{This is quite different from
the traditional dense cellular system, where coordination among only a limited number of BSs is usually sufficient for improving the cell-edge users' performance, e.g., the coordinated multipoint transmission and reception technology in LTE-Advanced standard~\cite{r18}, and the global coordination can be decoupled into local coordinations with an acceptable performance loss~\cite{r12}.}. Consequently, a scalable interference coordination scheme is desired.

Using only the large-scale CSI of interference links, a coordinated frequency resource block (FRB) allocation problem is formulated for maximizing the minimum achievable rate of all the users. Our analysis reveals that the optimization problem is NP-hard in general.
To circumvent this difficulty, a greedy
scheme with polynomial-time complexity is proposed by adopting the bisection method and linear integer programming tools.
Simulation results demonstrate that
interference coordination based on large-scale CSI only can still provide
substantial gains over the existing methods.
Moreover, although the proposed scheme is only guaranteed to converge to a local optimum,
its performance is not far from that of the global optimal solution, and
it performs well in terms of both user fairness and system efficiency.

The rest of this paper is organized as follows. Section II introduces the system model.
The problem of large-scale CSI based interference coordination is formulated and analyzed
in Section III, and
a low-complexity coordinated FRB allocation scheme is proposed in Section IV.
Section V shows the simulation results, and finally, conclusions are
given in Section VI.

Throughout this paper, lower case and upper case boldface symbols
denote vectors and matrices, respectively. $\mathcal{CN}(0, \sigma^2)$ denotes the
complex Gaussian distribution with zero mean and $\sigma^2$ variance.
$(\cdot)^H$ and $(\cdot)^T$ represent the transpose conjugate and the transpose, respectively.
$\mathbf{E}(\cdot)$ is the expectation operator.

\section{System Model}
As illustrated in Fig.~\ref{fig1}, we consider a two-tier network consisting
of a traditional microwave BS with wide-area coverage and a massive amount of
mmWave BSs~\cite{r20}. The microwave BS provides global control of the whole system, thus enabling the network-level interference coordination.
The mmWave BSs are densely deployed in an amorphous way, which support high-capacity transmission.
In practice, the deployment of mmWave BSs mainly depends on the geographical distribution of target users and environmental blockages, as well as candidate BS sites.
Similar to the traditional dense cellular system with distributed antennas~\cite{rfeng6}-\cite{rfeng8}, exclusive optical fibers are adopted in this model to connect the microwave BS and all the mmWave BSs\footnote{The system backhaul may be limited in some scenarios, e.g., when using the standardized interface named X2~\cite{r18}. Then, the impact of limited backhaul capability
should be taken into consideration, which is out of the scope of this paper, however is an interesting future research topic.}.
Although the mmWave bandwidth is large, it is not feasible to allocate orthogonal FRBs for all the BSs, considering
the ever-increasing user demand. We assume $N$ frequency division clusters (FDCs), each of which is formed by a group of neighboring mmWave BSs.
Given the association relationship~\cite{r015}, we consider
$K$ `BS$\rightarrow$user' pairs in each FDC, which are allocated with orthogonal FRBs,
to mitigate strong interferences\footnote{Note that $K$ is an important parameter for the system, which should be determined according to
the density of mmWave BSs, users' requirement, and available FRBs. So far the optimization of $K$ is still an open issue. Furthermore, given $K$, the optimal formation of FDCs is
also worth studying.}.
However, these $K$ FRBs are reused among different FDCs, causing co-channel interference.
We adopt a uniform linear array (ULA) with $N_a$ antenna elements and a single RF chain at each BS. Accordingly only the analog beamforming is required. Meanwhile,
to reduce the cost at user equipments and
support high user mobility, only a single omnidirectional antenna is assumed at each mobile terminal.

\begin{figure}[t]
\begin{center}
  \includegraphics[width=8.8cm]{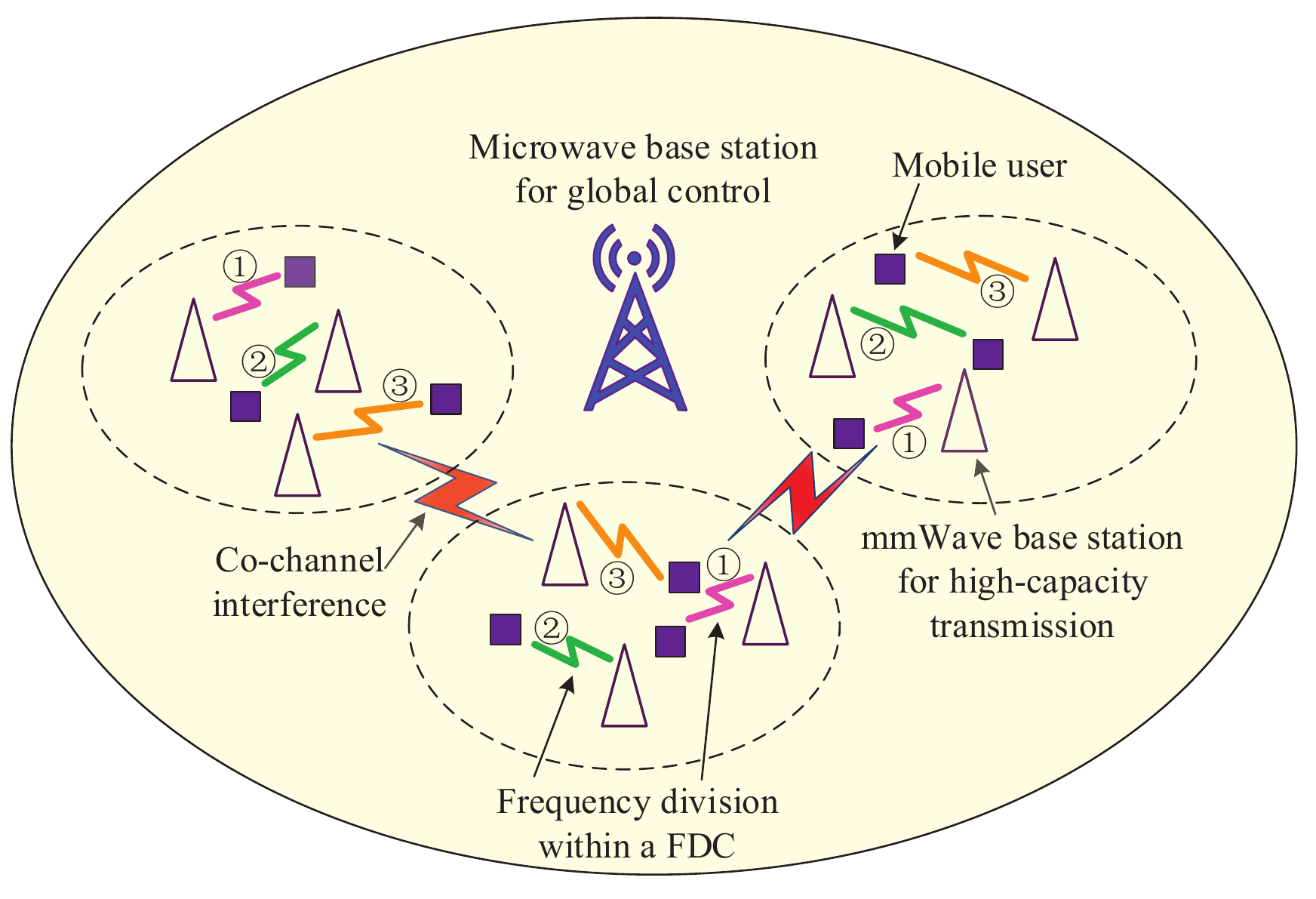}
  \caption{Illustration of an extremely dense and amorphous mmWave network.}
  \label{fig1}
\end{center}
\end{figure}

Let $\{z_{k,l}^{(n)} \, | \, z_{k,l}^{(n)}\in \{0,1\}, k, l=1,2,...,K, n = 1,2,...,N\}$
denote the FRB allocation indicators. If FRB $l$ is allocated to user $k$ in FDC $n$, then $z_{k,l}^{(n)}=1$; otherwise
$z_{k,l}^{(n)}=0$. To mitigate strong interferences within each FDC, the FRB allocation indicators satisfy
\begin{eqnarray}\label{eq:e2_1}
&&\sum_{l=1}^{K} z_{k,l}^{(n)}=1, \,\, \sum_{k=1}^{K} z_{k,l}^{(n)}=1,\nonumber \\
&&k, l=1,2,...,K, n = 1,2,...,N.
\end{eqnarray}

We focus on the downlink system and consider only the azimuth~\cite{r07}. The parameterized
geometric mmWave channel model with limited number of
scatterers is adopted~\cite{r07}. Without loss of generality, each scatterer contributes a single propagation path between the BS and user~\cite{r07}.
Accordingly, we can express the channel vector between BS $i$ in FDC $m$ and user $k$ in FDC $n$ in FRB $l$ as
\begin{eqnarray}\label{eq:add1}
&&\mathbf h_{l,k,i}^{(n,m)}=(\eta_{k,i}^{(n,m)})^{-1/2} \times \nonumber\\
&&\left[ \begin{array}{lll}
\sum_{\iota = 1}^{L_{k,i}^{(n,m)}} \alpha_{\iota,l,k,i}^{(n,m)} \\
\sum_{\iota = 1}^{L_{k,i}^{(n,m)}} \alpha_{\iota,l,k,i}^{(n,m)} e^{j(2\pi/\lambda)\tau \sin(\phi_{\iota,k,i}^{(n,m)})} \\
 \vdots  \\
\sum_{\iota = 1}^{L_{k,i}^{(n,m)}} \alpha_{\iota,l,k,i}^{(n,m)} e^{j(N_a-1)(2\pi/\lambda)\tau \sin(\phi_{\iota,k,i}^{(n,m)})}
\end{array} \right],
\end{eqnarray}
where $\eta_{k,i}^{(n,m)}$ denotes the path loss, $L_{k,i}^{(n,m)}$ represents the number of scatterers, $\alpha_{\iota,l,k,i}^{(n,m)}\sim \mathcal{CN}(0,1)$ denotes
the small-scale Rayleigh fading generated by path $\iota$, $\lambda$ is the signal wavelength, $\tau$ represents
the distance between antenna elements, and $\phi_{\iota,k,i}^{(n,m)}\in [0,2\pi]$ denotes the azimuth angle of departure (AoD) for path $\iota$.
For modeling the path loss, we further have
\begin{eqnarray}
\eta_{k,i}^{(n,m)}[dB] = \nu_{k,i}^{(n,m)} + \beta_{k,i}^{(n,m)} 10 \log_{10}(d_{k,i}^{(n,m)})+\xi_{k,i}^{(n,m)},
\end{eqnarray}
where $d_{k,i}^{(n,m)}$ is the distance, and $\xi_{k,i}^{(n,m)}\sim \mathcal{CN}(0,(o_{k,i}^{(n,m)})^2)$ denotes the lognormal shadowing. $\nu_{k,i}^{(n,m)}$, and $\beta_{k,i}^{(n,m)}$, as well as
$(o_{k,i}^{(n,m)})^2$ take different values according to the channel status, i.e., whether containing a line-of-sight (LOS) path, whether encountering a blockage, and so on~\cite{r011}.

We categorize these channel parameters into large-scale and small-scale ones, as shown in Table~\ref{tab1}. To control the cost for CSI acquisition, we assume only the large-scale CSI is available
for all the interference links. Consequently, the network-level interference coordination can make use of the large-scale CSI only.

\begin{table}
  \centering
  \caption{Large-scale CSI versus small-scale CSI in the geometric mmWave channel model.}
  \label{tab1}
\begin{tabular}{|l|l|l|}
  \hline
    \textbf{Category} & \textbf{Parameters} & \textbf{Property}\\
  \hline
  \emph{Large-scale } & $\bullet$ Path loss: $\eta_{k,i}^{(n,m)}$   & Varies slowly on\\
   CSI & $\bullet$ Number of scatterers: $L_{k,i}^{(n,m)}$ & the order of user's\\
   & $\bullet$ AoDs: $\phi_{\iota,k,i}^{(n,m)}$ & location change\\
  \hline
  \emph{Small-scale } & $\bullet$ Rayleigh fading: $\alpha_{\iota,l,k,i}^{(n,m)}$ & Varies quickly \\
  CSI &  & on the order of \\
   &  & a signal wavelength\\
  \hline
\end{tabular}
\end{table}

\section{Large-Scale CSI based Interference Coordination}
\subsection{Problem Formulation}
For FRB $l$ in FDC $n$, the received signal can be expressed as
\begin{eqnarray}\label{eq:e2_2}
&&\!\!\!\!\!\!\!\!\!\!\!\!\!\!\!\!\!\! y_l^{(n)} = \sum_{k=1}^{K}  (\mathbf h_{l,k,k}^{(n,n)})^T \mathbf w_{l,k,k}^{(n,n)} x_{k}^{(n)} z_{k,l}^{(n)} \nonumber \\
&&\!\!\!\!\!\!\!\!\!\!\!\!\!\!\!\!\!\! + \sum_{k=1}^{K} \left( \sum_{m=1,m\neq n}^{N}   \sum_{i=1}^{K} (\mathbf h_{l,k,i}^{(n,m)})^T \mathbf w_{l,i,i}^{(m,m)} x_{i}^{(m)} z_{i,l}^{(m)} \right) z_{k,l}^{(n)} \nonumber \\
&&\!\!\!\!\!\!\!\!\!\!\!\!\!\!\!\!\!\!+ n_l^{(n)},
\end{eqnarray}
where $\mathbf w_{l,k,k}^{(n,n)}$ denotes the phase-shifting vector at BS $k$,
$x_{i}^{(m)}$ is the transmitted signal for user $i$ in FDC $m$,
and $n_l^{(n)}$ is the white Gaussian noise with distribution $\mathcal{CN}(0,\sigma ^2)$.
Denoting the transmit power of each BS as $P$, which is equally allocated to the front $N_a$ antenna elements, we have
\begin{eqnarray}
\mathbf{E} (|x_{i}^{(m)}|^2)=\frac{P}{N_a},~~i=1,2,...,K, ~m = 1,2,...,N.
\end{eqnarray}
To maximize the link margin of each `BS$\rightarrow$user' pair, we set
\begin{eqnarray}\label{eq:add2}
\mathbf w_{l,k,k}^{(n,n)}=\left[ \begin{array}{lll}
\frac{(\sum_{\iota = 1}^{L_{k,k}^{(n,n)}} \alpha_{\iota,l,k,k}^{(n,n)})^H}{|\sum_{\iota = 1}^{L_{k,k}^{(n,n)}} \alpha_{\iota,l,k,k}^{(n,n)}|}\\
\frac{(\sum_{\iota = 1}^{L_{k,k}^{(n,n)}} \alpha_{\iota,l,k,k}^{(n,n)} e^{j(2\pi/\lambda)\tau \sin(\phi_{\iota,k,k}^{(n,n)})})^H}{|\sum_{\iota = 1}^{L_{k,k}^{(n,n)}} \alpha_{\iota,l,k,k}^{(n,n)} e^{j(2\pi/\lambda)\tau \sin(\phi_{\iota,k,k}^{(n,n)})}|} \\
 \vdots  \\
\frac{(\sum_{\iota = 1}^{L_{k,k}^{(n,n)}} \alpha_{\iota,l,k,k}^{(n,n)} e^{j(N_a-1)(2\pi/\lambda)\tau \sin(\phi_{\iota,k,k}^{(n,n)})})^H}{|\sum_{\iota = 1}^{L_{k,k}^{(n,n)}} \alpha_{\iota,l,k,k}^{(n,n)} e^{j(N_a-1)(2\pi/\lambda)\tau \sin(\phi_{\iota,k,k}^{(n,n)})}|}
\end{array} \right]
\end{eqnarray}
each entry of which corresponds to the phase shifting value for analog beamforming, such that the signals from different antenna elements are aligned and accordingly the energy is focused on the corresponding user.

For the transmission link, we define
\begin{eqnarray}
&&g_{l,k,k}^{(n,n)} =  |(\mathbf h_{l,k,k}^{(n,n)})^T \mathbf w_{l,k,k}^{(n,n)}|^2,\nonumber \\
&&l, k=1,2,...,K, n = 1,2,...,N.
\end{eqnarray}
For the interference link, we define (\ref{eq:heng}) on the top of the next page.
\begin{figure*}[!t]
\begin{eqnarray}\label{eq:heng}
&&\!\!\!\!\!\!\!\!\! g_{l,k,i}^{(n,m)} = \mathbf{E}|(\mathbf h_{l,k,i}^{(n,m)})^T \mathbf w_{l,i,i}^{(m,m)}|^2  \nonumber \\
&&\!\!\!\!\!\!\!\!\! \overset{(\vartriangle)}{=}\frac{1}{\eta_{k,i}^{(n,m)}} \mathbf{E}|\sum_{n_a = 1}^{N_a} \left[\frac{(\sum_{\ell = 1}^{L_{i,i}^{(m,m)}} \alpha_{\ell,l,i,i}^{(m,m)} e^{j(n_a-1)(2\pi/\lambda)\tau \sin(\phi_{\ell,i,i}^{(m,m)})})^H}{|\sum_{\ell = 1}^{L_{i,i}^{(m,m)}} \alpha_{\ell,l,i,i}^{(m,m)} e^{j(n_a-1)(2\pi/\lambda)\tau \sin(\phi_{\ell,i,i}^{(m,m)})}|} \sum_{\iota = 1}^{L_{k,i}^{(n,m)}} \alpha_{\iota,l,k,i}^{(n,m)} e^{j(2(n_a-1)\pi/\lambda)\tau \sin(\phi_{\iota,k,i}^{(n,m)})}\right]|^2 \nonumber \\
&&\!\!\!\!\!\!\!\!\! \overset{(\triangledown)}{=} \frac{1}{\eta_{k,i}^{(n,m)}} \sum_{\iota = 1}^{L_{k,i}^{(n,m)}} |\sum_{n_a = 1}^{N_a} \frac{(\sum_{\ell = 1}^{L_{i,i}^{(m,m)}} \alpha_{\ell,l,i,i}^{(m,m)} e^{j(n_a-1)(2\pi/\lambda)\tau \sin(\phi_{\ell,i,i}^{(m,m)})})^H}{|\sum_{\ell = 1}^{L_{i,i}^{(m,m)}} \alpha_{\ell,l,i,i}^{(m,m)} e^{j(n_a-1)(2\pi/\lambda)\tau \sin(\phi_{\ell,i,i}^{(m,m)})}|}e^{j(2(n_a-1)\pi/\lambda)\tau \sin(\phi_{\iota,k,i}^{(n,m)})}|^2, \nonumber \\
&&\!\!\!\!\!\!\!\!\!l, k, i=1,2,...,K, n, m = 1,2,...,N.
\end{eqnarray}
\end{figure*}
In (\ref{eq:heng}), the expectation is taken over the unknown small-scale CSI, i.e., $\alpha_{\iota,l,k,i}^{(n,m)},~n\neq m$, to obtain the statistical channel gain of interference links. Particularly, the derivation ($\vartriangle$) can be obtained from (\ref{eq:add1}) and (\ref{eq:add2}), and the derivation ($\triangledown$)
can be obtained by leveraging the independent Gaussian distribution of $\alpha_{\iota,l,k,i}^{(n,m)},~n\neq m$.
Then based on (\ref{eq:e2_2}), the received signal to interference and noise ratio (SINR) in FRB $l$ in FDC $n$ can be derived as
\begin{eqnarray}
&&\!\!\!\!\!\!\!\!\!\bar{S}_l^{(n)} = \nonumber \\
&&\!\!\!\!\!\!\!\!\!\frac{\sum_{k=1}^{K} g_{l,k,k}^{(n,n)} z_{k,l}^{(n)}}{ \sum_{k=1}^{K} \left( \sum_{m\neq n}   \sum_{i=1}^{K} g_{l,k,i}^{(n,m)} z_{i,l}^{(m)} \right) z_{k,l}^{(n)} + N_a\sigma^2/P },
\end{eqnarray}
where the total interference is regarded to be Gaussian~\cite{r13}\cite{r14}, according to the central limit theorem.

In order to control the interference, and guarantee the transmission rate of all the users, we formulate the following optimization problem
\begin{subequations}
\begin{align}
& \underset{z_{k,l}^{(n)}} \max \,\, \underset{l,n} \min \,\,\,\, \log_2(1+\bar{S}_l^{(n)}) \\
& s.t. \,\, \sum_{l=1}^{K} z_{k,l}^{(n)}=1, \sum_{k=1}^{K} z_{k,l}^{(n)}=1, z_{k,l}^{(n)}\in \{0,1\}, \\
& ~~~~~~\!k,l=1,..., K, \,\,n=1,...,N,
\end{align}
\end{subequations}
where the minimum of the achievable rate of all the users is maximized, by optimizing the FRB allocation strategy.
This is particularly significant when all the users require similar quality of service (QoS).
According to
the monotonically-increasing property of the $\log$ function, we can equivalently simplify the problem as
\begin{subequations}\label{eq:e3_3}
\begin{align}
& \underset{z_{k,l}^{(n)}} \max \,\, \underset{l,n} \min \,\,\,\, \bar{S}_l^{(n)} \\
& s.t. \,\, \sum_{l=1}^{K} z_{k,l}^{(n)}=1, \sum_{k=1}^{K} z_{k,l}^{(n)}=1, z_{k,l}^{(n)}\in \{0,1\}, \\
&~~~~~~\!k,l=1,..., K, \,\,n=1,...,N.
\end{align}
\end{subequations}

\subsection{Problem Analysis}
Because of the non-linearity of the objective function,
the coordinated FRB allocation problem in~(\ref{eq:e3_3}) is
a complex nonlinear integer programming problem that
is rather difficult to solve.
Theorem~\ref{Theorem_NP} shows that~(\ref{eq:e3_3}) is NP-hard for any fixed $N \geq 3$ with respect to
the number of users $K$ in each FDC.

\begin{theorem}\label{Theorem_NP}
The coordinated FRB allocation problem in~(\ref{eq:e3_3}) can be solved in polynomial time when $N = 2$, and it is NP-hard with respect to $K$
for any fixed $N \geq 3$.
\end{theorem}

\begin{proof}
See Appendix A.
\end{proof}

The basic idea of the proof can be summarized as follows.
When $N = 2$, the coordinated FRB allocation problem in~(\ref{eq:e3_3}) can be transformed into a series of linear bottleneck assignment
problems (LBAPs) through bisection searching.
Because an LBAP is polynomial-time solvable~\cite{r22}\cite{r23}, we know that (\ref{eq:e3_3}) can be solved in polynomial time when $N=2$.
For any fixed $N \geq 3$, we can present a polynomial reduction from a three-dimensional axial bottleneck assignment
problem (3BAP) to~(\ref{eq:e3_3}).
Based on the fact that a 3BAP problem is NP-hard~\cite{r24}-\cite{r26}, we can conclude that~(\ref{eq:e3_3}) is NP-hard when $N \geq 3$.

\section{A Greedy Coordinated FRB Allocation Scheme}
Theorem~\ref{Theorem_NP} indicates that
for a dense mmWave network with more than three FDCs and a large number of users to serve in each FDC,
it will be computation-consuming to find an optimal solution for
the coordinated FRB allocation problem in~(\ref{eq:e3_3}).
In this section,
a low-complexity algorithm is designed to solve~(\ref{eq:e3_3}) in a greedy way.
Particularly, a degraded problem derived from~(\ref{eq:e3_3}) is investigated in Section IV-A to offer insights into the original problem,
based on which a greedy coordinated FRB allocation scheme is presented in Section IV-B.

\subsection{A Degraded Case}
Since the problem in (\ref{eq:e3_3}) is difficult to solve directly, we first investigate a degraded case to understand the original problem.
We consider a system consisting of $U$ FDCs, and assume that the
FRB allocation strategies for the first $U-1$ FDCs, i.e., $\{z_{k,l}^{(n)} \, | \, k,l=1,...,K, n=1,2,...,U-1\}$, have been
given and only the FRB allocation strategy for FDC $U$ needs to be optimized.
Note that all of the system parameters in the degraded case are the same as the original system, and the only difference
lies in the reduced optimization variables.
In this case, $\bar{S}_l^{(n)}$ can be simplified as
\begin{eqnarray}\label{eq:e4_1}
&&\!\!\!\!\!\!\!\!\!\bar{S}_l^{(n)} = \nonumber \\
&&\!\!\!\!\!\!\!\!\!\frac{ g_{l,u_l^n,u_l^n}^{(n,n)}}{ \sum_{m=1,m\neq n}^{U-1} g_{l,u_l^n,u_l^m}^{(n,m)} +  \sum_{k=1}^{K} g_{l,u_l^n,k}^{(n,U)} z_{k,l}^{(U)} + N_a\sigma^2/P },\nonumber \\
&&\!\!\!\!\!\!\!\!\!l=1,...,K, n=1,..., U-1,
\end{eqnarray}
and
\begin{eqnarray}\label{eq:e4_2}
&&\!\!\!\!\!\!\!\!\!\!\!\!\!\!\!\!\!\!\bar{S}_l^{(U)} = \frac{\sum_{k=1}^{K} g_{l,k,k}^{(U,U)} z_{k,l}^{(U)} }{ \sum_{k=1}^{K} \left( \sum_{m=1}^{U-1} g_{l,k,u_l^m}^{(U,m)} \right) z_{k,l}^{(U)} + N_a\sigma^2/P },\nonumber \\
&&\!\!\!\!\!\!\!\!\!\!\!\!\!\!\!\!\!\!l=1,...,K,
\end{eqnarray}
where $u_l^n$ denotes the user scheduled in FRB $l$ in FDC $n,n=1,2,...,U-1$.
Then, the coordinated FRB allocation problem can be written as
\begin{subequations}\label{eq:e4_3}
\begin{align}
& \underset{z_{k,l}^{(U)}} \max \,\, \underset{l,n} \min \,\,\,\, \bar{S}_l^{(n)} \\
& s.t. \,\, \sum_{l=1}^{K} z_{k,l}^{(U)}=1, \sum_{k=1}^{K} z_{k,l}^{(U)}=1, z_{k,l}^{(U)}\in \{0,1\}, \\
&~~~~~~\!k,l=1,..., K,~n=1,..., U,
\end{align}
\end{subequations}
which is equivalent to
\begin{subequations}\label{eq:e4_4}
\begin{align}
& \max \,\,t \\
& s.t. \,\, \sum_{l=1}^{K} z_{k,l}^{(U)}=1, \sum_{k=1}^{K} z_{k,l}^{(U)}=1, z_{k,l}^{(U)}\in \{0,1\},   \\
& \,\,\,\,\,\,\,\,\,\,\,\, \bar{S}_l^{(n)} \geq t,\,\, k,l=1,..., K, \,\,n=1,..., U,
\end{align}
\end{subequations}
by introducing a slack variable $t$.

Note that (\ref{eq:e4_4}) is also a nonlinear integer programming problem
due to the non-linearity of the constraints in~(\ref{eq:e4_4}c).
One way to obtain the optimal solution is exhaustive search, but
the complexity is factorial of the number of users in each FDC, i.e., $K!$.
In the following, we propose an algorithm that can solve~(\ref{eq:e4_4}) in polynomial time.
Particularly,
for any fixed $t$, the non-linear constraints in~(\ref{eq:e4_4}c) can be transformed
into a set of linear constraints
\begin{eqnarray}\label{eq:e4_5}
\sum_{k=1}^{K} a_{k,l}^{n} z_{k,l}^{(U)} \leq b_{l}^{n} ,\,\, k,l=1,..., K, \,\,n=1,..., U,
\end{eqnarray}
with
\begin{eqnarray}
&&\!\!\!\!\!\!\!\!\!\!\!\!\!\!\!\!\!\!\!\! a_{k,l}^{n}=t  g_{l,u_l^n,k}^{(n,U)}, \\
&&\!\!\!\!\!\!\!\!\!\!\!\!\!\!\!\!\!\!\!\! b_{l}^{n}=g_{l,u_l^n,u_l^n}^{(n,n)} - t  \left( \sum_{m=1,m\neq n}^{U-1} g_{l,u_l^n,u_l^m}^{(n,m)}  +  N_a\sigma^2/P \right),\\
&&\!\!\!\!\!\!\!\!\!\!\!\!\!\!\!\!\!\!\!\! n=1,2,...,U-1, \nonumber
\end{eqnarray}
and
\begin{eqnarray}
&& a_{k,l}^{U} = t \sum_{m=1}^{U-1} g_{l,k,u_l^m}^{(U,m)} -  g_{l,k,k}^{(U,U)}, \\
&& b_{l}^{U} = - t N_a \sigma^2/P.
\end{eqnarray}
Then the problem in (\ref{eq:e4_4}) can be decomposed into
a series of feasibility linear integer programming problems by using the bisection method~\cite{r27}.
Table~\ref{tb:Algorithm} presents the detailed algorithm, where
the feasibility linear integer programming problem that needs to be solved in each iteration
can be written as
\begin{subequations}\label{eq:e4_8}
\begin{align}
& \text{find} \,\, \{ z_{k,l}^{(U)} \}_{k,l=1,...,K} \\
& s.t. \,\, \sum_{l=1}^{K} z_{k,l}^{(U)}=1, \sum_{k=1}^{K} z_{k,l}^{(U)}=1, z_{k,l}^{(U)}\in \{0,1\},  \\
& \,\,\,\,\,\,\,\,\,\, \sum_{k=1}^{K} \hat{a}_{k,l}^{n} z_{k,l}^{(U)} \leq \hat{b}_{l}^{n} ,\,\, k,l=1,..., K, \,\,n=1,..., U,
\end{align}
\end{subequations}
where
\begin{eqnarray}
&&\!\!\!\!\!\!\!\!\!\!\!\!\!\!\!\!\!\!\!\! \hat{a}_{k,l}^{n}=t_m g_{l,u_l^n,k}^{(n,U)},\\
&&\!\!\!\!\!\!\!\!\!\!\!\!\!\!\!\!\!\!\!\! \hat{b}_{l}^{n}=g_{l,u_l^n,u_l^n}^{(n,n)} - t_m  \left( \sum_{m=1,m\neq n}^{U-1} g_{l,u_l^n,u_l^m}^{(n,m)}  +  N_a\sigma^2/P \right), \\
&&\!\!\!\!\!\!\!\!\!\!\!\!\!\!\!\!\!\!\!\! n=1,2,...,U-1, \nonumber
\end{eqnarray}
and
\begin{eqnarray}
&& \hat{a}_{k,l}^{U} = t_m \sum_{m=1}^{U-1} g_{l,k,u_l^m}^{(U,m)} -  g_{l,k,k}^{(U,U)}, \\
&& \hat{b}_{l}^{U} = - t_m N_a \sigma^2/P.
\end{eqnarray}
Using the maximum and minimum interference from other FDCs for each user,
$v_{min}$ and $v_{max}$, i.e.,
the lower bound and the upper bound of $t$,
can be set as
\begin{subequations}\label{eq:e4_9}
\begin{align}
v_{min} = \underset{n,k}\min \left\{ \frac{ g_{l,k,k}^{(n,n)}}{ \sum_{m\neq n} \text{max}_{i} \left( g_{l,k,i}^{(n,m)} \right) + N_a\sigma^2/P } \right\}, \\
v_{max} = \underset{n,k}\max  \left\{ \frac{ g_{l,k,k}^{(n,n)}}{ \sum_{m\neq n} \text{min}_{i} \left( g_{l,k,i}^{(n,m)} \right) + N_a\sigma^2/P } \right\},
\end{align}
\end{subequations}
respectively.

\begin{table}
  \centering
  \caption{The algorithm for solving (\ref{eq:e4_4}).}
  \label{tb:Algorithm}
  \begin{tabular}{l}
    \hline \hline
$\!\!\!\!$\emph{Initialization:}~ Set $ U^0 = v_{max}, L^0 = v_{min}$, and $\epsilon = 1.0\times 10^{-3}$.\\
    \hline
$\!\!\!\!$\emph{Iterations:} $s=1,2,\ldots$\\
     \quad 1) Let $t_m = (U^{s-1}+L^{s-1})/2$.\\
     \quad 2) Solve the feasibility linear integer programming problem in (\ref{eq:e4_8}).  \\
     \quad 3) If the problem is feasible, $U^{s} = U^{s-1}, L^{s} = t_m$; otherwise \\
     \quad \quad  $ U^{s} = t_m, L^{s} = L^{s-1}$. \\
     \quad 4) If $|U^{s}-L^{s}|/L^{s}<\epsilon$, stop; otherwise go to 1). \\
    \hline
   \end{tabular}
\end{table}

The following
Lemma~\ref{Bi_LBAP} shows that~(\ref{eq:e4_8}) can be solved in $O(K^{2.5} / \sqrt{\text{log} (K)})$ time,
which is polynomial with respect to the number of users $K$ in each FDC.
As the problem
in~(\ref{eq:e4_4}) can be solved through solving a series of feasibility problem as that in~(\ref{eq:e4_8}) based on bisection searching,
we can conclude that~(\ref{eq:e4_4}) can be solved in polynomial time~\cite{r28}.

\begin{lemma}\label{Bi_LBAP}
The feasibility problem in~(\ref{eq:e4_8}) can be transformed equivalently into an LBAP
as shown in~(\ref{eq:e4_10})
and it can be solved in $O(K^{2.5} / \sqrt{\text{log} (K)})$ time.
\begin{subequations}\label{eq:e4_10}
\begin{align}
& \underset{z_{k,l}^{(U)}} \min \,\, \underset{k,l} \max \,\,\,\,   \hat{c}_{k,l} z_{k,l}^{(U)} \\
& s.t. \,\, \sum_{l=1}^{K} z_{k,l}^{(U)}=1, \sum_{k=1}^{K} z_{k,l}^{(U)}=1, z_{k,l}^{(U)}\in \{0,1\}, \\
&~~~~~~\! k,l=1,..., K, \,\,n=1,..., U,
\end{align}
\end{subequations}
where
\begin{eqnarray}
\hat{c}_{k,l} =  \max_{n} \left\{(M + \hat{a}_{k,l}^n) / (M + \hat{b}_{l}^n) \right\},
\end{eqnarray}
and $M$ is a positive constant satisfying
\begin{eqnarray}
M + \min \left\{ \min_{k,l,n} \{ \hat{a}_{k,l}^n \}, \, \min_{l,n} \{ \hat{b}_{l}^n \} \right\} > 0.
\end{eqnarray}
\end{lemma}

\begin{proof}
See Appendix B.
\end{proof}

\begin{figure}[t]
\begin{center}
  \includegraphics[width=6.7cm]{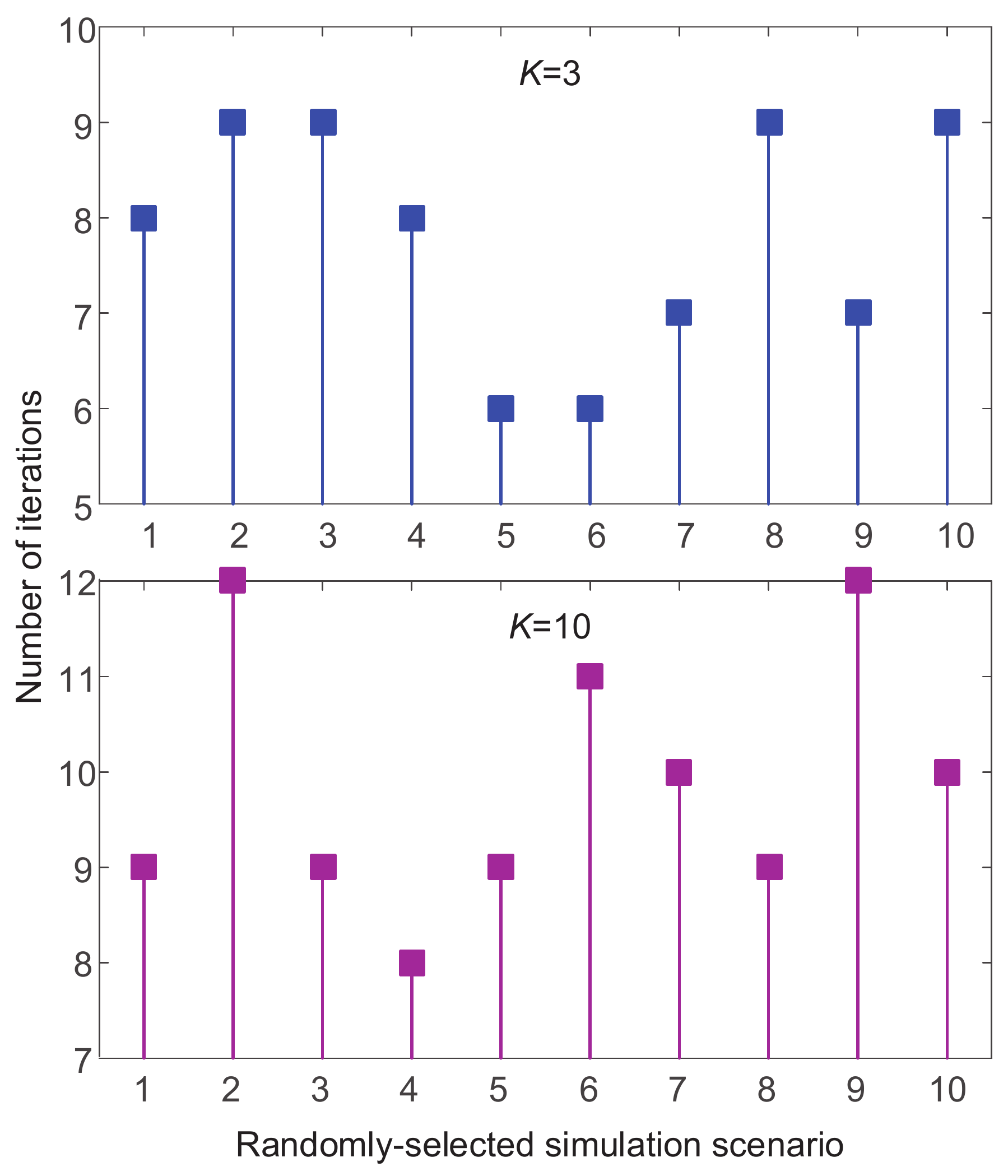}
  \caption{Number of iterations needed by the algorithm in Table~\ref{tb:Algorithm}.}
  \label{fig_iter_1}
\end{center}
\end{figure}

With the same simulation setting as that used in Section V,
the number of iterations needed by the bisection searching algorithm in Table~\ref{tb:Algorithm}
is shown in Fig.~\ref{fig_iter_1}, for randomly-selected 10 out of 1000 simulation scenarios.
We can see that for both $K=3$ and $K=10$,
it takes around $10$ iterations to converge.

\subsection{A Greedy Solution}
Note that if we optimize the FRB allocation indicators FDC by FDC and fix the FRB allocation strategies of all the other FDCs
at each time when the FRB allocation of one FDC is being optimized,
then we can transform the original coordinated FRB allocation
problem in (\ref{eq:e3_3}) into a series of much simpler subproblems as shown in (\ref{eq:e4_4}).
Obviously, this approach may quite possibly fail to obtain the global optimal result. However, it uncovers a potential pathway to find a local optimum solution of the problem.
Based on this observation,
we propose a greedy coordinated FRB allocation scheme for solving~(\ref{eq:e3_3}).

\begin{table}
  \centering
  \caption{The proposed coordinated FRB allocation scheme.}
  \label{tb:Algorithm_all}
  \begin{tabular}{l}
    \hline \hline
$\!\!\!\!$\emph{Initialization:}~ \\
     \quad $\mathcal{A} = \{g(1)\}$,\\
     \quad $z_{k,k}^{(g(1))}=1$, $z_{k,l}^{(g(1))}=0$ for $l\neq k,\,\,k,l=1,..., K$.\\
    \hline
$\!\!\!\!$\emph{Step 1:} \emph{for} $s=2,3,...,N,$\\
     \quad \quad \quad $\,$ 1) Add FDC $g(s)$ to $\mathcal{A}$ and set $U = s$.\\
     \quad \quad \quad $\,$ 2) Solve (\ref{eq:e4_4}) for the subsystem $\mathcal{A}$ and obtain \\
     \quad \quad \quad \quad $\,\,$ the FRB allocation indicators $z_{k,l}^{(g(U))},\,\,k,l=1,..., K$.  \\
$\!\!\!\!$     \quad \quad \quad \emph{end} \\
$\!\!\!\!$\emph{Step 2:} Set the FRB allocation indicators of the $N$ FDCs \\
$\!\!\!\!$     \quad \quad \quad according to the results achieved in Step 1,  \\
$\!\!\!\!$     \quad \quad \quad and denote the minimum of all $\bar{S}_l^{(n)}$ by $S_0$.\\
$\!\!\!\!$     \quad \quad \quad Let $U = N$, $\epsilon = 1.0 \times 10^{-3}$. \\
$\!\!\!\!$\emph{Step 3: Iterations:} $r=1,2,...$\\
     \quad \quad 1) \emph{for} $s=1,2,...,N,$ \\
     \quad \quad \quad \quad $\,$ Take FDC $g(s)$ as the $U$th FDC and the other $N-1$ \\
     \quad \quad \quad \quad $\,$ FDCs as FDC $1 \sim U-1$. Solve the problem in (\ref{eq:e4_4}) \\
     \quad \quad \quad \quad $\,$ and update the FRB allocation indicators for FDC $g(s)$. \\
     \quad \quad \quad \emph{end} \\
     \quad \quad 2) Calculate $\bar{S}_l^{(n)}$, $n=1,...,N,l=1,...,K$, according to \\
     \quad \quad \quad the updated FRB allocation indicators in 1), \\
     \quad \quad \quad and denote the minimum of all $\bar{S}_l^{(n)}$ by $S_r$. \\
          \quad \quad 3) If $| S_r - S_{r-1} | / S_{r-1} < \epsilon$, stop; otherwise go to 1).\\
     \quad \quad \quad \\
    \hline
   \end{tabular}
\end{table}

As illustrated in Table~\ref{tb:Algorithm_all},
the proposed coordinated FRB allocation scheme consists of three steps.
In Step $1$, $N$ FDCs are added to a subsystem $\mathcal{A}$ one by one
according to the sequence determined by a function $f$
(the function $g$ in Table~\ref{tb:Algorithm_all} is the inverse function of $f$, i.e., $g = f^{-1}$).
Accordingly,
we optimize the FRB allocation of the $N$ FDCs one by one
by tackling a subproblem as that in (\ref{eq:e4_4}) with $U=s$ in each iteration $s$ ($s=1,...,N$).
Based on the FRB allocation results of the $N$ FDCs determined in Step $1$,
we iteratively optimize the FRB allocation of each FDC with the others fixed in Step $2$ and Step $3$
so as to further improve the performance of the system.
In each iteration in Step 3,
we need to solve $N$ subproblems as shown in (\ref{eq:e4_4}) with $U=N$.

In Table~\ref{tb:Algorithm_all},
the function $f:\{1,2,...,N\}\mapsto\{1,2,...,N\}$ is used to reorder the $N$ FDCs in the system
and determine the sequence of the FDCs in the optimization,
so as to provide better fairness among all the users.
In the proposed scheme, we reorder the FDCs according to the total interference power they suffer.
The total interference power of FDC $n$, $n=1,2,...,N$, can be calculated as
\begin{eqnarray}\label{eq:e4_11}
PoI_n = \sum_{k=1}^{K}\sum_{m=1,m\neq n}^{N}\sum_{i=1}^{K}\frac{g_{l,k,i}^{(n,m)}P}{N_a}.
\end{eqnarray}
The function $f$ renumbers the FDC with the largest interference power as $1$ and the FDC with the smallest
interference power as $N$.
It means that the FDCs suffering more severe interference are optimized
earlier so that
the FDCs whose performances are severely affected by co-channel interference
have higher priority in the optimization. Intuitively,
this is consistent with the fairness-oriented optimization target. In theory, there should be an optimal function $f$ for the proposed scheme, which will be
studied in our future work.

In Table~\ref{tb:Algorithm_all}, it is clear that $S_r$ is nondecreasing with $r$, and is upper bounded by the optimal value of the objective function
of~(\ref{eq:e3_3}), thus, the proposed scheme is assured
to converge.
From Table~\ref{tb:Algorithm} and Table~\ref{tb:Algorithm_all},
we can see that the proposed scheme actually transforms~(\ref{eq:e3_3}) into
a series of LBAPs as that shown in~(\ref{eq:e4_10}),
and thus it has a polynomial-time computational complexity.
Fig.~\ref{fig_iter_2} illustrates the convergence speed of the proposed scheme
under the simulation setting in Section V.
It can be seen that the proposed scheme converges rather quickly.
For both $K=3$ and $K=10$,
it needs less than $5$ iterations to converge.

\begin{figure}[t]
\begin{center}
  \includegraphics[width=6.7cm]{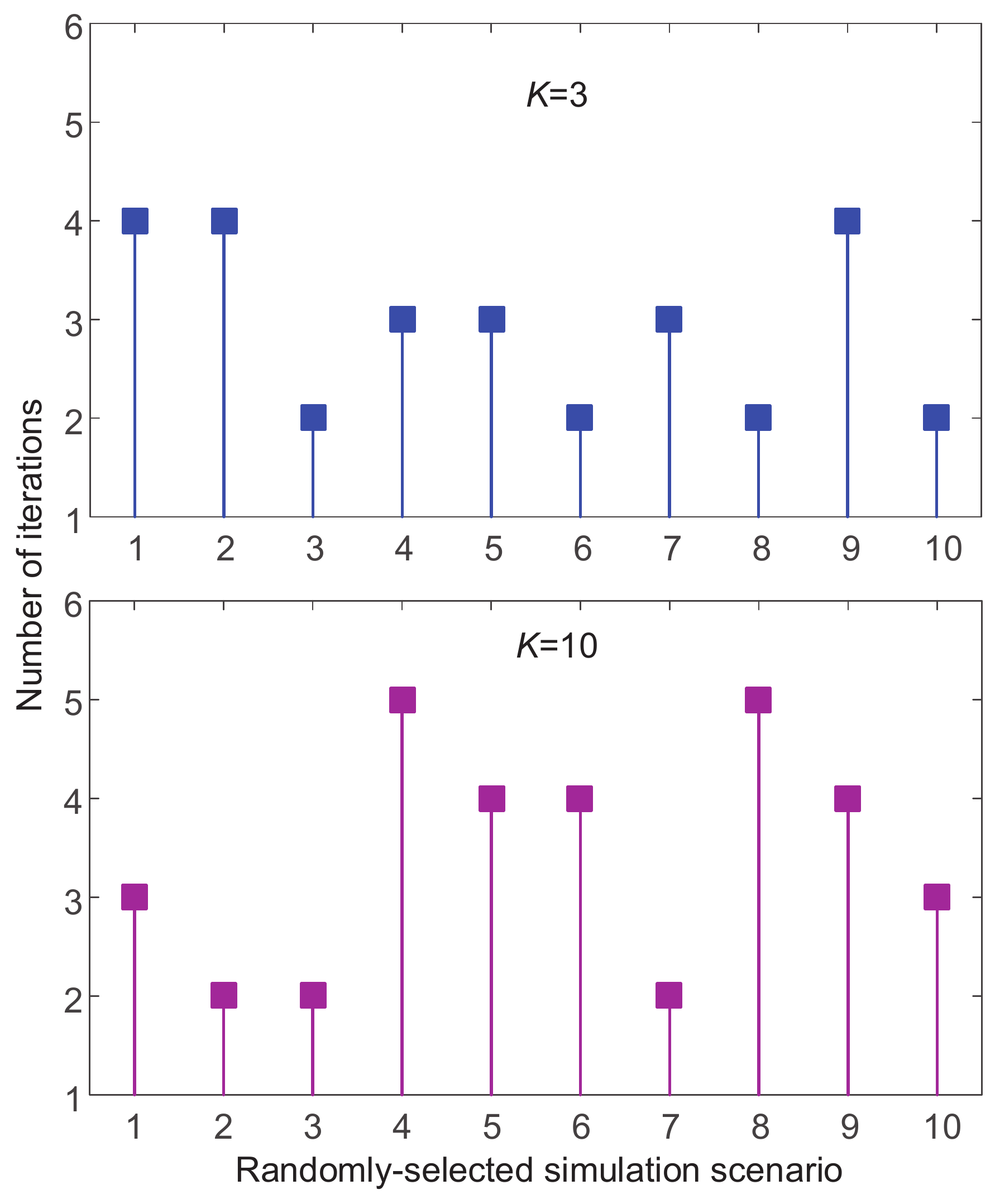}
  \caption{Number of iterations needed by the scheme in Table~\ref{tb:Algorithm_all}.}
  \label{fig_iter_2}
\end{center}
\end{figure}

In practical applications, the proposed scheme can be implemented in a periodic way, so as to adapt to the network dynamics. However, from the perspective of reducing cost, it can also be triggered only when the total co-channel interference is more severe than a given threshold. In this case, a real-time interference estimator is required to monitor the status of the network.

\section{Simulation Results}
In this part, we demonstrate the performance of the proposed scheme by comparing it with
the theoretically optimal result obtained by exhaustive search, the state-of-the-art method designed with full CSI~\cite{r15}, and the simplest single-FDC scheme with no interference coordination.
Note that the scheme in~\cite{r15} was designed to maximize the sum rate of a dense cellular network. In
the traditional cellular system, the performance of different users associated with different cells is statistically similar due to the cellular-like regular network topology.
Therefore, sum rate maximization is a reasonable optimization goal. However, when network densification meets mmWave communications, an amorphous system topology
is inevitable in practice as discussed in Section I. Thereby, minimum rate maximization becomes a more urgent optimization objective, as discussed in Section III.

In our simulations, the microwave BS is located at the center of a circular coverage area, and the mmWave BSs are deployed following the two-dimensional
uniform distribution.
We consider $N=10$, $K=3~\text{or}~10$, $N_a = 16$. The system is operated at the 28GHz carrier frequency.
At each mmWave BS, the distance between antenna elements is set as $\tau = \frac{\lambda}{2}$. As for the large-scale channel parameters,
we set
\begin{eqnarray}
\eta_{k,i}^{(n,m)}[dB] = 61.4 + 20 \log_{10}(d_{k,i}^{(n,m)})+\xi_{k,i}^{(n,m)},
\end{eqnarray}
for the LOS case with $\xi_{k,i}^{(n,m)} \sim \mathcal{CN}(0,5.8)$, and
\begin{eqnarray}
\eta_{k,i}^{(n,m)}[dB] = 72.0 + 29.2 \log_{10}(d_{k,i}^{(n,m)})+\xi_{k,i}^{(n,m)},
\end{eqnarray}
for the non-LOS case with $\xi_{k,i}^{(n,m)} \sim \mathcal{CN}(0,8.7)$~\cite{r011}. Without loss of generality, the number of scatterers are assumed
to be 3 for all the transmission and interference links, and all the AoDs are randomly generalized in $[0,2\pi]$ with an uniform distribution.
We consider $10$ randomly-generalized large-scale CSI conditions with different path loss and different AoDs, and 100 small-scale Rayleigh fading realizations
for each large-scale CSI condition. In the sequel, all the results are derived by averaging over these 1000 simulation scenarios.

\begin{figure}[t]
\begin{center}
  \includegraphics[width=8.1cm]{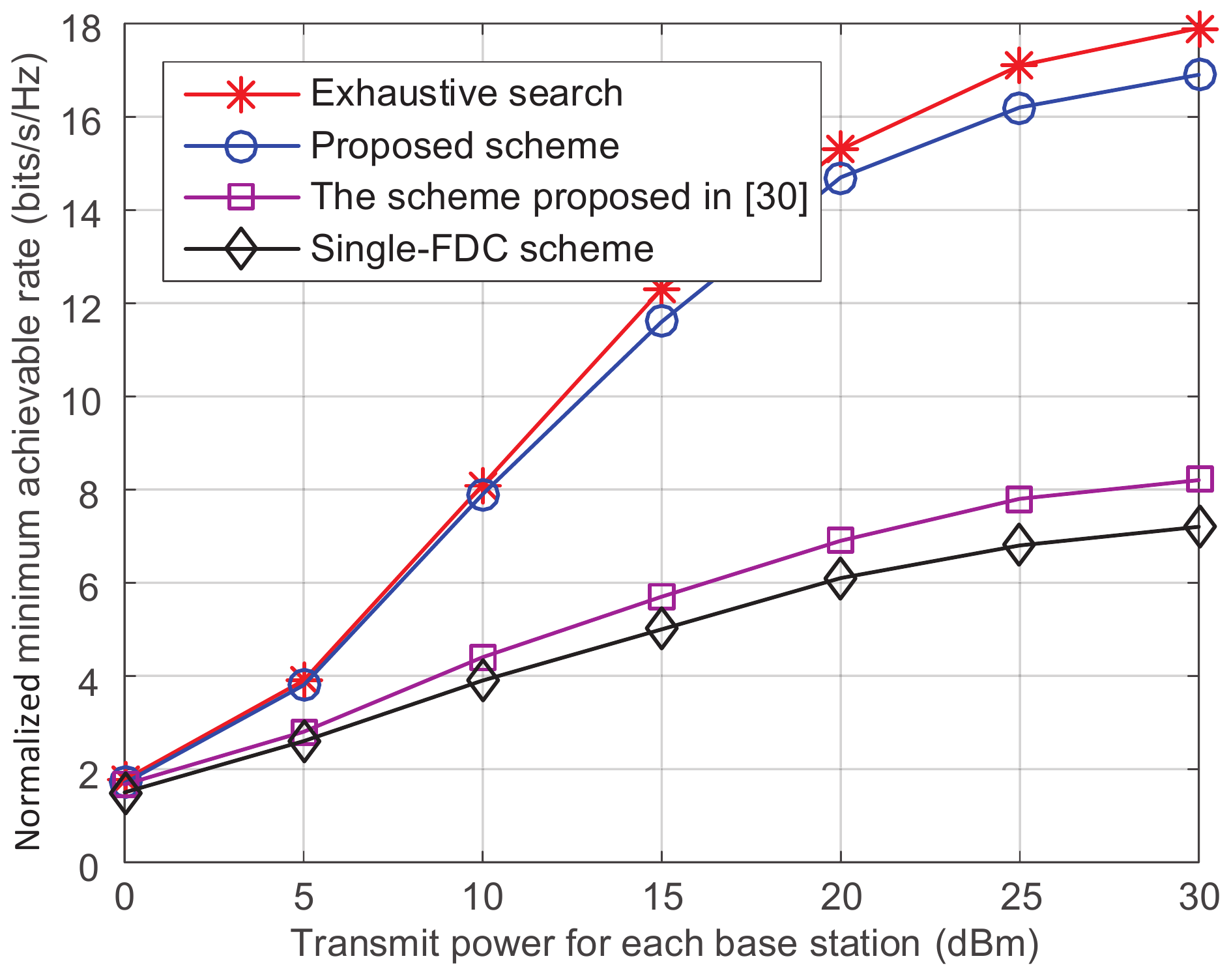}
  \caption{Minimum normalized achievable rate of all the users for $K=3$.}
  \label{fig_min_1}
\end{center}
\end{figure}

We start from the simpler scenario with $K=3$, for which the exhaustive search can be used
to find the optimal FRB allocation strategy based on~(\ref{eq:e3_3}). Fig.~\ref{fig_min_1}
compares the minimum achievable rate (normalized by the bandwidth) by different FRB allocation schemes.
It can be clearly observed that although the proposed scheme is only guaranteed to converge to a local optimum,
its performance is not far from that of the global optimum.
Recall that the proposed scheme optimizes the FRB allocation indicators repeatedly
through the iterative operation in Step $3$ of Table~\ref{tb:Algorithm_all}.
By optimizing the sequence of the FDCs
via the function $f$, a higher priority is given to those users who suffer larger interference in the optimization.
Accordingly, the minimum achievable rate can be significantly improved.
Moreover, a remarkable performance gap can be seen in the figure between the proposed scheme and
the scheme proposed in~\cite{r15}, as well as the single-FDC scheme.
Notably the gain goes larger as the transmit power of each mmWave BS increases,
which can be explained by the fact that the more severe the co-channel interference, the higher the performance gain that
can be achieved by interference coordination.

With the same density of mmWave BSs, Fig.~\ref{fig_min_2} further illustrates the performance of the proposed scheme for a larger number of users in each FDC, i.e., $K=10$.
In this case, exhaustive search for the optimal FRB allocation strategy is not considered as it is
too time-consuming. We can still observe a significant gain offered by the proposed scheme over the other methods.
On one hand, the performance of the scheme proposed in~\cite{r15} is almost the same as that of the single-FDC scheme, which
indicates that maximizing merely the sum rate can not guarantee to improve the minimum achievable rate of the users, under an amorphous system topology. On the other hand,
if only the large-scale CSI is available,
the performance of the scheme proposed in~\cite{r15}, which was originally designed on the basis of full CSI, could be severely degraded.
By comparing Fig.~\ref{fig_min_1} and Fig.~\ref{fig_min_2}, we find that the minimum achievable rate for $K=3$ is a little better
than that for $K=10$ under the same condition, which can be explained by the fact that the worst user channel becomes worse as the number of users increases, with the same density of mmWave BSs.

\begin{figure}[t]
\begin{center}
  \includegraphics[width=8.1cm]{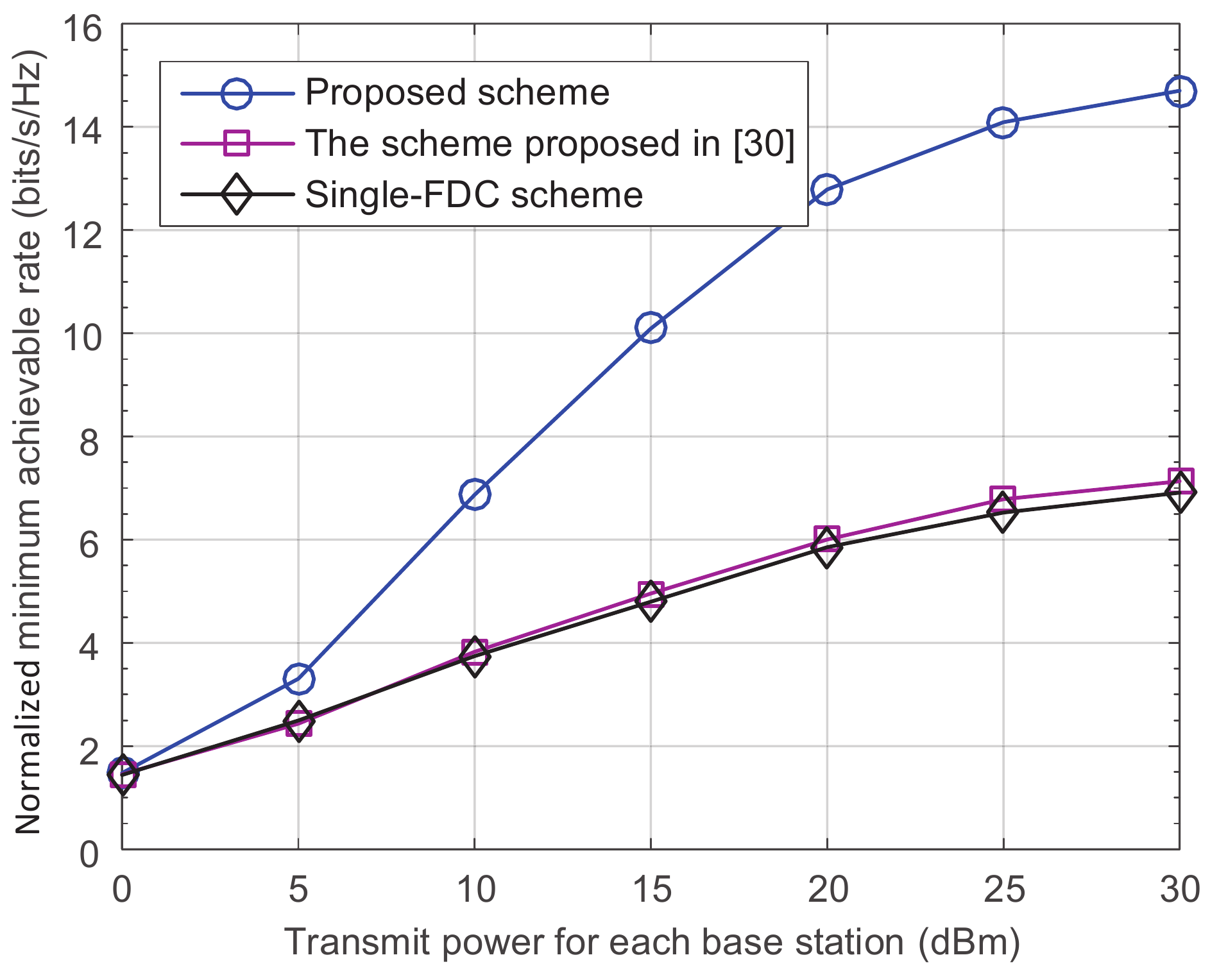}
  \caption{Minimum normalized achievable rate of all the users for $K=10$.}
  \label{fig_min_2}
\end{center}
\end{figure}

We also evaluate the performance of the proposed scheme in terms of system sum rate, i.e., $\sum_{l=1}^K \sum_{n=1}^N \log_2(1+\bar{S}_l^{(n)})$, in Fig.~\ref{fig_sum} for $K=3$. Particularly, the global orthogonal FRB allocation scheme is taken into comparison,
which does not produce any co-channel interference, but at the expense of $N$ times frequency resource usage. We can observe that
frequency reuse among different FDCs is effective for improving the total spectral efficiency of the system.
It can also be seen from the figure that the proposed scheme provides a noticeable performance gain in system efficiency in addition to
user fairness.
Although the scheme in~\cite{r15} focused on the maximization of system sum rate, it was designed based on full CSI, and its performance could be severely degraded when only the large-scale CSI is available. As for the single-FDC scheme, it directly ignored the impact of co-channel interference, leading to a largely-reduced system efficiency.
Basically, the proposed scheme only focuses on promoting the achievable rate of the poorest user, which
will inevitably damage the performance of some related users, although it may simultaneously improve the achievable rate of some other users.
Therefore, the gain of the proposed scheme in terms of system sum rate is not proportional to that in terms of minimum achievable rate.

\begin{figure}[t]
\begin{center}
  \includegraphics[width=8.1cm]{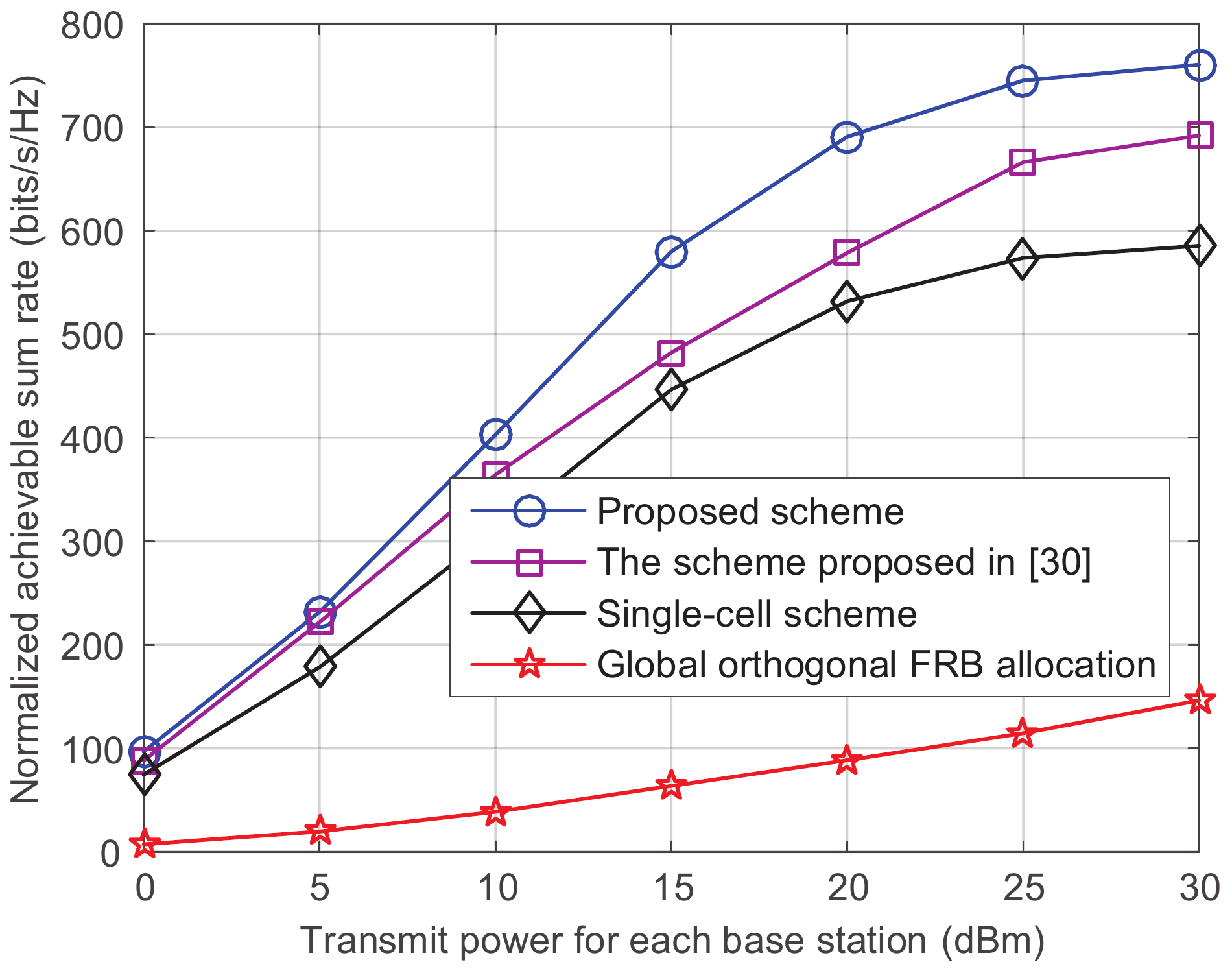}
  \caption{Normalized achievable sum rate for $K=3$.}
  \label{fig_sum}
\end{center}
\end{figure}

\section{Conclusions}
In this paper, we have explored the potential of network densification for providing
a sufficient link margin for mmWave communications. We propose to deploy the mmWave BSs
in an extremely dense and amorphous way, so as to reduce the access distance and give more
serving BS choice for each mobile user. Under this model, the propagation loss and blockages, which are the most challenging hurdles
for mmWave transmissions, can be flexibly mitigated. However, the
co-channel interference becomes a performance-limiting factor. To solve this problem, we have
proposed a large-scale CSI based coordinated FRB allocation scheme, which can maximize the minimum
achievable rate of all the users. Since the large-scale CSI varies slowly and is highly location-dependent, we can
obtain it with a quite low system cost. Therefore, the scalability of the proposed coordination
can be guaranteed, which is crucial for the considered hard-to-decouple amorphous mmWave network.
Simulation results demonstrate that
coordinated FRB allocation based on large-scale CSI only
can still significantly improve the system performance.
Moreover, although the proposed scheme is only guaranteed to converge to a local optimum,
its performance is not far from that of the global optimal solution,
and it performs well in terms of both user fairness and system efficiency.

Undoubtedly, the proposed model requires more BS sites, and leads to the other extra cost, e.g., synchronization among multiple BSs.
Therefore, on the basis of this work, more practical limitation of BS sites as well as more comprehensive cost models can be
taken into consideration, and the practically-optimal BS density of a mmWave network can be investigated in the future.

\appendices
\section{Proof of Theorem~\ref{Theorem_NP}}
When $N=2$, to
solve the problem in~(\ref{eq:e3_3}), we need to match the users into $K$ pairs, each of which
consists of one user in FDC 1 and one user in FDC 2, so that
the minimum of $\bar{S}_l^{(n)}, l=1,..., K,~n=1,...,N,$ can be maximized.
It indicates that we can find an optimal solution of~(\ref{eq:e3_3})
by scheduling the $K$ users in FDC $1$ in the $K$ FRBs in an arbitrary order
and then optimizing the scheduling order of FDC $2$ based on the scheduling order of FDC $1$
according to the objective function of~(\ref{eq:e3_3}).
Based on the analysis in Section IV-A, we can see that~(\ref{eq:e3_3})
can be transformed into a series of LBAPs
through the bisection method and solved in polynomial time
when $N=2$.

Now let us look at the case when $N \geq 3$.

By introducing a slack variable $\gamma$, the problem in (\ref{eq:e3_3}) can be equivalently rewritten as
\begin{subequations}\label{eq:eB_1}
\begin{align}
& \max \,\, \gamma  \\
& s.t. \,\, \sum_{l=1}^{K} z_{k,l}^{(n)}=1, \sum_{k=1}^{K} z_{k,l}^{(n)}=1, z_{k,l}^{(n)}\in \{0,1\},  \\
& ~~~~~~ \bar{S}_l^{(n)} \geq \gamma, \,\, k,l=1,..., K, \,\,n=1,...,N,
\end{align}
\end{subequations}
which can be solved through solving a series of
feasibility subproblems as
\begin{subequations}\label{eq:eB_2}
\begin{align}
& \text{find} \,\, z_{k,l}^{(n)} \\
& s.t. \,\, \sum_{l=1}^{K} z_{k,l}^{(n)}=1, \sum_{k=1}^{K} z_{k,l}^{(n)}=1, z_{k,l}^{(n)}\in \{0,1\}, \\
& ~~~~~~ \bar{S}_l^{(n)} \geq \gamma_0, \,\, k,l=1,..., K, \,\,n=1,...,N,
\end{align}
\end{subequations}
through bisection searching,
just like that in Table~\ref{tb:Algorithm}.
Conversely, if an optimal solution of~(\ref{eq:eB_1}), denoted by $\{ \gamma^\ast, z_{k,l}^{\ast(n)}| k,l=1,...,K,n=1,...,N \}$, is given,
then the problem in~(\ref{eq:eB_2}) can be solved immediately.
If $\gamma^\ast \geq \gamma_0$, then~(\ref{eq:eB_2}) is feasible and $\{ z_{k,l}^{\ast(n)} | k,l=1,...,K,n=1,...,N \}$
is a solution of~(\ref{eq:eB_2}); otherwise, (\ref{eq:eB_2}) is infeasible.
It indicates that~(\ref{eq:eB_2}) can be solved through solving the problem in~(\ref{eq:eB_1}),
and further it can be solved by implementing any algorithm for solving the problem in~(\ref{eq:e3_3}).
Thus, (\ref{eq:eB_2}) is polynomially reducible to (\ref{eq:e3_3})~\cite{r29}.

Further, (\ref{eq:eB_2}) can be equivalently transformed into
\begin{subequations}\label{eq:eB_3}
\begin{align}
&  \text{find} \,\, z_{k,l}^{(n)}  \\
&  s.t. \,\, \sum_{l=1}^{K} z_{k,l}^{(n)}=1, \sum_{k=1}^{K} z_{k,l}^{(n)}=1, z_{k,l}^{(n)}\in \{0,1\},  \\
&  ~~~~~ \gamma_0 \sum_{k=1}^{K} \sum_{m \neq n} \sum_{i=1}^{K} g_{l,k,i}^{(n,m)} z_{i,l}^{(m)} z_{k,l}^{(n)} - \sum_{k=1}^{K} g_{l,k,k}^{(n,n)} z_{k,l}^{(n)}  \nonumber \\
&  ~~~~~\leq - \gamma_0 \frac{N_a\sigma^2}{P}, \,\, k,l=1,..., K, \,\,n=1,...,N.
\end{align}
\end{subequations}
By adding $M \sum_{k=1}^{K} z_{k,l}^{(n)}$ (note that $M \sum_{k=1}^{K} z_{k,l}^{(n)} = M$) and $M$ ( $M > \gamma_0 \frac{N_a\sigma^2}{P}$ ) at the left-hand side and right-hand side of the inequality constraints
in~(\ref{eq:eB_3}c)
respectively,  we can rewrite~(\ref{eq:eB_3}) as
\begin{subequations}\label{eq:eB_4}
\begin{align}
& \text{find} \,\, z_{k,l}^{(n)}  \\
& s.t. \,\, \sum_{l=1}^{K} z_{k,l}^{(n)}=1, \sum_{k=1}^{K} z_{k,l}^{(n)}=1, z_{k,l}^{(n)}\in \{0,1\},  \\
& ~~~~~  \sum_{k=1}^{K} \sum_{m \neq n} \sum_{i=1}^{K} \hat{g}_{l,k,i}^{(n,m)} z_{i,l}^{(m)} z_{k,l}^{(n)} + \sum_{k=1}^{K} \hat{g}_{l,k,k}^{(n,n)} z_{k,l}^{(n)} \nonumber \\
& ~~~~~  \leq 1,  k,l=1,..., K, \,\,n=1,...,N,
\end{align}
\end{subequations}
with
$\hat{g}_{l,k,i}^{(n,m)} = \gamma_0 g_{l,k,i}^{(n,m)} / (M - \gamma_0 \frac{N_a\sigma^2}{P}) \geq 0$
and $\hat{g}_{l,k,k}^{(n,n)} = (M - g_{l,k,k}^{(n,n)})/ (M - \gamma_0 \frac{N_a\sigma^2}{P})$.
Further, we can solve~(\ref{eq:eB_4}) through solving the following problem
\begin{subequations}\label{eq:eB_6}
\begin{align}
& \underset{z_{k,l}^{(n)}} \min \,\, \underset{n,l} \max \,\,\,\,   \sum_{k=1}^{K} \left ( \sum_{m \neq n} \sum_{i=1}^{K} \hat{g}_{l,k,i}^{(n,m)} z_{i,l}^{(m)} +  \hat{g}_{l,k,k}^{(n,n)} \right ) z_{k,l}^{(n)}  \\
& s.t. \,\, \sum_{l=1}^{K} z_{k,l}^{(n)}=1, \sum_{k=1}^{K} z_{k,l}^{(n)}=1,  z_{k,l}^{(n)}\in \{0,1\}, \\
& ~~~~~~ k,l=1,..., K, \,\,n=1,..., N.
\end{align}
\end{subequations}
If the minimum value of the objective function of~(\ref{eq:eB_6}) is less than or equal to $1$,
then~(\ref{eq:eB_4}) is feasible, i.e., we can find a set of $\{ z_{k,l}^{(n)} | k,l=1,...,K, n=1,...,N \}$
satisfying the constraints in~(\ref{eq:eB_4}b) and~(\ref{eq:eB_4}c);
otherwise~(\ref{eq:eB_4}) is infeasible.
Moreover, from the transformations in~(\ref{eq:eB_1}) and~(\ref{eq:eB_2}) and the related analysis,
it is easy to see that~(\ref{eq:eB_6}) can be solved through
solving a series of subproblems as in~(\ref{eq:eB_4}) through bisection searching.
As the transformations from~(\ref{eq:eB_2}) to~(\ref{eq:eB_4}) can be reversed easily,
we can conclude that the problem in~(\ref{eq:eB_6}) can be solved
by implementing any algorithm for solving~(\ref{eq:eB_2}) for a series of times.
Thus, the problem in~(\ref{eq:eB_6}) is polynomially reducible to~(\ref{eq:eB_2}),
and further by the transitivity of polynomial reduction, it is polynomially reducible to~(\ref{eq:e3_3})~\cite{r28}\cite{r29}.

When $N=3$,
we first
prove~(\ref{eq:eB_6}) is NP-hard by presenting a polynomial reduction of
a NP-hard 3BAP to it.
Then, based on the fact that (\ref{eq:eB_6}) is polynomially reducible to (\ref{eq:e3_3}),
we conclude that (\ref{eq:e3_3}) is also a NP-hard problem.
This is because if~(\ref{eq:e3_3}) is polynomially solvable, then
there must be a polynomial-time algorithm for~(\ref{eq:eB_6})~\cite{r29},
and this contradicts with the fact that~(\ref{eq:eB_6}) is NP-hard.

Let us consider a 3BAP as follows. Given three disjoint sets
of size $K$, i.e., $\mathcal{F}_1$, $\mathcal{F}_2$, and $\mathcal{F}_3$, we are to find a subset $\mathcal{G} \subseteq  \mathcal{F}_1 \times \mathcal{F}_2 \times \mathcal{F}_3$
such that $\mathcal{G}$ has $K$ triples with every element in $\mathcal{F}_1 \cup \mathcal{F}_2 \cup \mathcal{F}_3$
occurring in exactly one triple of $\mathcal{G}$, and that the maximum cost $A_{ijk}$ of all the triples $(i,j,k) \in \mathcal{G}$
is minimized. $A_{ijk}$ is defined as $A_{ijk} = \max \{ d_{ij}^{12}, d_{ik}^{13}, d_{jk}^{23}\}$, where $d_{ij}^{mn}$
corresponds to some kind of distance between $i \in  \mathcal{F}_m$ and $j \in  \mathcal{F}_n$ and $d_{ij}^{mn} \geq 0$.
This 3BAP can be formulated as
\begin{subequations}\label{eq:eB_7}
\begin{align}
& \underset{t_{ijk} } \min \,\, \underset{i,j,k} \max \,\,\,\,  A_{ijk} t_{ijk}    \\
& s.t. \,\, \sum_{i=1}^{K} \sum_{j=1}^{K} t_{ijk} = 1, \\
& ~~~~~ \sum_{i=1}^{K} \sum_{k=1}^{K} t_{ijk}=1, \\
& ~~~~~ \sum_{j=1}^{K} \sum_{k=1}^{K} t_{ijk}=1, \\
& ~~~~~~\! t_{ijk} \in \{0,1\}, i,j,k=1,..., K.
\end{align}
\end{subequations}
It has been shown that it is NP-hard~\cite{r24}.
Lemma \ref{lemma_NP} in Appendix C shows that when $N=3$, (\ref{eq:eB_7}) is polynomially reducible to (\ref{eq:eB_6}), which means that
if there is a polynomial-time algorithm for (\ref{eq:eB_6}), then there is also a polynomial-time algorithm for (\ref{eq:eB_7})~\cite{r29}.
Because (\ref{eq:eB_7}) is NP-hard, it is easy to see that (\ref{eq:eB_6}) is also NP-hard when $N=3$.
Thus, the problem in~(\ref{eq:e3_3}) is NP-hard when $N=3$.

For any fixed $N > 3$, we first consider a special case of~(\ref{eq:e3_3}), where
$g_{l,i,j}^{(n,m)}, \, g_{l,i,j}^{(m,n)}$,
and $g_{l,i,j}^{(m_{1},m_{2})}$ are all assumed to be $0$
for any $i,j =1,...,K,n=1,2,3$, and $m,m_1,m_2=4,...,N$, $m_{1} \neq m_{2}$.
In this special case, there is no ICI between the first three FDCs
and the other $N-3$ FDCs, and also there is no ICI between
any two FDCs whose FDC numbers are both larger than $3$.
It is easy to see that in this case, we only need to optimize the FRB allocation indicators for the first three FDCs,
and this degenerates to the case when $N=3$.
This indicates that the problem in~(\ref{eq:e3_3}) with $N=3$ is a special case of that with $N > 3$,
i.e., the problem in~(\ref{eq:e3_3}) with $N=3$ can be solved directly by any algorithm for that with $N>3$.
Thus, the problem in~(\ref{eq:e3_3}) with $N=3$ is polynomially reducible to that with any fixed $N > 3$.
Now we can conclude that the problem~(\ref{eq:e3_3}) is also NP-hard when $N > 3$.

\section{Proof of Lemma~\ref{Bi_LBAP}}
Suppose $M$ is a large enough positive real number so that
\begin{eqnarray}
M + \min \left\{ \min_{k,l,n} \{ \hat{a}_{k,l}^n \}, \, \min_{l,n} \{ \hat{b}_{l}^n \} \right\} > 0.
\end{eqnarray}
By adding $M \sum_{k=1}^{K} z_{k,l}^{(U)} $ and $M$ (note that $M \sum_{k=1}^{K} z_{k,l}^{(U)} = M$) to the left-hand side and right-hand side
of the inequality constraints in~(\ref{eq:e4_8}c), respectively,
we can transform~(\ref{eq:e4_8}) equivalently into
\begin{subequations}\label{eq:eD_1}
\begin{align}
& \text{find} \,\, \{ z_{k,l}^{(U)} \}_{k,l=1,...,K} \\
& s.t. \,\, \sum_{l=1}^{K} z_{k,l}^{(U)}=1, \sum_{k=1}^{K} z_{k,l}^{(U)}=1, z_{k,l}^{(U)}\in \{0,1\},  \\
& \,\,\,\,\,\,\,\,\,\, \sum_{k=1}^{K} c_{k,l}^{n} z_{k,l}^{(U)} \leq 1  ,\,\, k,l=1,..., K, \,\,n=1,..., U,
\end{align}
\end{subequations}
with $c_{k,l}^{n} =  (M + \hat{a}_{k,l}^n) / (M + \hat{b}_{l}^n) > 0$.
Further, we can solve~(\ref{eq:eD_1}) by solving the following problem
\begin{subequations}\label{eq:eD_2}
\begin{align}
& \underset{z_{k,l}^{(U)}} \min \,\, \underset{l,n} \max \,\,\,\,  \sum_{k=1}^{K} c_{k,l}^{n} z_{k,l}^{(U)}  \\
& s.t. \,\, \sum_{l=1}^{K} z_{k,l}^{(U)}=1, \sum_{k=1}^{K} z_{k,l}^{(U)}=1, z_{k,l}^{(U)}\in \{0,1\}, \\
& ~~~~~~\!  k,l=1,..., K, \,\,n=1,..., U.
\end{align}
\end{subequations}
If the minimum value of the objective function of~(\ref{eq:eD_2}) is less than or equal to $1$,
then~(\ref{eq:eD_1}) is feasible, i.e., we can find a set of $\{ z_{k,l}^{(U)} | k,l=1,...,K \}$
satisfying the constraints in~(\ref{eq:eD_1}b) and~(\ref{eq:eD_1}c);
otherwise~(\ref{eq:eD_1}) is infeasible.

Because $c_{k,l}^{n} > 0$ and only one of $z_{k,l}^{(U)}, k=1,...,K$, is 1 with others equal to 0,
we have
\begin{eqnarray}\label{eq:eD_3}
\underset{l,n} \max \,\,  \sum_{k=1}^{K} c_{k,l}^{n} z_{k,l}^{(U)} = \underset{k,l,n} \max \,\,   c_{k,l}^{n} z_{k,l}^{(U)}
= \underset{k,l} \max \,\,   \hat{c}_{k,l} z_{k,l}^{(U)},
\end{eqnarray}
with
\begin{eqnarray}
\hat{c}_{k,l} = \max_{n} c_{k,l}^{n}.
\end{eqnarray}
Thus, we can further transform~(\ref{eq:eD_2}) into an LBAP
\begin{subequations}\label{eq:eD_5}
\begin{align}
& \underset{z_{k,l}^{(U)}} \min \,\, \underset{k,l} \max \,\,\,\,   \hat{c}_{k,l} z_{k,l}^{(U)} \\
& s.t. \,\, \sum_{l=1}^{K} z_{k,l}^{(U)}=1, \sum_{k=1}^{K} z_{k,l}^{(U)}=1, z_{k,l}^{(U)}\in \{0,1\}, \\
& ~~~~~~\! k,l=1,..., K, \,\,n=1,..., U,
\end{align}
\end{subequations}
which is the same with that in~(\ref{eq:e4_10}).

The LBAP in~(\ref{eq:eD_5}) has a $K \times K$ cost matrix $\{ \hat{c}_{k,l} \}_{k,l=1,...,K}$.
From the Theorem $6.4$ in~\cite{r22}, we know that~(\ref{eq:eD_5}) can be solved in
$O(K^{2.5} / \sqrt{\text{log} (K)})$ time.
Thus, we can conclude that the problem in~(\ref{eq:e4_8}) can be solved in $O(K^{2.5} / \sqrt{\text{log} (K)})$ time.

\section{Lemma~\ref{lemma_NP}}
\begin{lemma} \label{lemma_NP}
The problem in (\ref{eq:eB_7}) is polynomially reducible to (\ref{eq:eB_6}) when $N=3$.
\end{lemma}

\begin{proof}
First we present another formulation for the problem in (\ref{eq:eB_6}) when $N=3$.
In fact, (\ref{eq:eB_6}) can be considered as minimizing the maximum cost of the users
corresponding to different FRB allocation results.
If user $i$ in FDC $1$, user $j$ in FDC $2$, and user $k$
in FDC $3$ are scheduled in the same FRB $l$, then according
to the objective function of (\ref{eq:eB_6}),
the costs of users $i,j,k$ are given by
$\hat{g}_{l,i,j}^{(1,2)} + \hat{g}_{l,i,k}^{(1,3)} + \hat{g}_{l,i,i}^{(1,1)}$, $\hat{g}_{l,j,i}^{(2,1)} + \hat{g}_{l,j,k}^{(2,3)} + \hat{g}_{l,j,j}^{(2,2)}$,
and $\hat{g}_{l,k,i}^{(3,1)} + \hat{g}_{l,k,j}^{(3,2)} + \hat{g}_{l,k,k}^{(3,3)}$, respectively.
To find an optimal solution of (\ref{eq:eB_6}), we
need to match the users into $K$ triples, each of which
consists of one user in FDC 1, one user in FDC 2, and one user
in FDC 3, so that
the maximum cost of the users is minimized.
Assuming that $\{ z_{k,l}^{\ast(n) } | k,l=1,...,K, n=1,2,3 \}$ is an optimal solution of (\ref{eq:eB_6})
and the users $i,j,k$ scheduled in the same FRB forms a tripe $(i,j,k)$,
then for any other feasible solution $\{ z_{k,l}^{(n)} | k,l=1,...,K, n=1,2,3 \}$ of (\ref{eq:eB_6}), as long as
the $K$ triples corresponding to it are the same with those corresponding to $\{ z_{k,l}^{\ast(n) } | k,l=1,...,K, n=1,2,3 \}$,
it is also an optimal solution.
Based on the above analysis, we can solve (\ref{eq:eB_6}) through solving
a 3BAP as
\begin{subequations}\label{eq:eC_1}
\begin{align}
& \underset{t_{ijk} } \min \,\, \underset{i,j,k} \max \,\,\,\,  B_{ijk} t_{ijk}   \\
& s.t. \,\, \sum_{i=1}^{K} \sum_{j=1}^{K} t_{ijk} = 1, \\
& ~~~~~ \sum_{i=1}^{K} \sum_{k=1}^{K} t_{ijk}=1, \\
& ~~~~~ \sum_{j=1}^{K} \sum_{k=1}^{K} t_{ijk}=1, \\
& ~~~~~~\! t_{ijk} \in \{0,1\}, i,j,k=1,..., K,
\end{align}
\end{subequations}
where $B_{ijk} = \max \{ \hat{g}_{l,i,j}^{(1,2)} + \hat{g}_{l,i,k}^{(1,3)} + \hat{g}_{l,i,i}^{(1,1)}, \, \hat{g}_{l,j,i}^{(2,1)} + \hat{g}_{l,j,k}^{(2,3)} + \hat{g}_{l,j,j}^{(2,2)}, \, \hat{g}_{l,k,i}^{(3,1)} + \hat{g}_{l,k,j}^{(3,2)} + \hat{g}_{l,k,k}^{(3,3)} \}$.
If an optimal solution of (\ref{eq:eB_6}) is given,
we can easily find an optimal solution for (\ref{eq:eC_1}) by setting $t_{i,j,k}=1$
as long as the users $i,j,k$ are scheduled in the same FRB.
Thus, the problem in~(\ref{eq:eC_1}) can be solved by any algorithm for~(\ref{eq:eB_6}).
Consequently, we know that (\ref{eq:eC_1}) is polynomially reducible to~(\ref{eq:eB_6}).

Next we show that (\ref{eq:eB_7}) is polynomially reducible to (\ref{eq:eC_1}),
and then by the transitivity of polynomial-time reduction, we can claim that
(\ref{eq:eB_7}) is polynomially reducible to (\ref{eq:eB_6}).

Notice that if $\hat{g}_{l,i,k}^{(1,3)}, \, \hat{g}_{l,j,i}^{(2,1)}, \, \hat{g}_{l,k,j}^{(3,2)}$,
and $\hat{g}_{l,k,k}^{(1,1)}, \, \hat{g}_{l,k,k}^{(2,2)}, \, \hat{g}_{l,k,k}^{(3,3)}$
are all equal to $0$
and we define
\begin{eqnarray}\label{eq:eApp_B-2}
d_{ij}^{12} \triangleq \hat{g}_{l,i,j}^{(1,2)}, \, d_{jk}^{23} \triangleq \hat{g}_{l,j,k}^{(2,3)}, \,
d_{ik}^{13} \triangleq \hat{g}_{l,k,i}^{(3,1)},
\end{eqnarray}
then the problem in (\ref{eq:eC_1}) is the same with (\ref{eq:eB_7}).
This means that (\ref{eq:eB_7}) is a special case of (\ref{eq:eC_1}).
Thus, if there is an algorithm for solving (\ref{eq:eC_1}),
then we can use it to solve (\ref{eq:eB_7}) directly,
which indicates that (\ref{eq:eB_7}) is polynomially reducible to (\ref{eq:eC_1}).
\end{proof}

\end{document}